\documentclass[conference]{IEEEtran}
\IEEEoverridecommandlockouts

\usepackage{cite}
\usepackage{adjustbox}
\usepackage{pdfpages}
\usepackage{amssymb,lineno}
\usepackage{algorithm}
\usepackage{algorithmic}
\usepackage{multirow}
\usepackage[tight,footnotesize]{subfigure}
\usepackage{epsfig}
\usepackage{footmisc}
\usepackage{footnpag}
\usepackage{amsmath}
\usepackage{amsthm}
\usepackage{paralist}
\usepackage{color}
\usepackage{perpage}
\usepackage{verbatim}
\usepackage{enumitem}
\usepackage{tikz}
\usepackage{listings}
\usepackage{graphicx}
\usepackage{soul}
\usepackage{fancyhdr}
\usepackage{circledtext}

\newcommand*\circled[1]{\tikz[baseline=(char.base)]{
            \node[shape=circle,fill,inner sep=0.5pt] (char) {\textcolor{white}{#1}};}}

\def\BibTeX{{\rm B\kern-.05em{\sc i\kern-.025em b}\kern-.08em
    T\kern-.1667em\lower.7ex\hbox{E}\kern-.125emX}}
    
\setlist[itemize]{topsep=\parskip}

\newtheorem{theorem}{Theorem}

\begin{document}

\pagestyle{plain}

\title{HoSZp: An Efficient Homomorphic Error-bounded
Lossy Compressor for Scientific Data\vspace{-5mm}}

\author{\IEEEauthorblockN{Tripti Agarwal\IEEEauthorrefmark{1},
Sheng Di\IEEEauthorrefmark{2},
Jiajun Huang\IEEEauthorrefmark{3},
Yafan Huang \IEEEauthorrefmark{4},
Ganesh Gopalakrishnan\IEEEauthorrefmark{1},\\
Robert Underwood\IEEEauthorrefmark{2},
Kai Zhao\IEEEauthorrefmark{5},
Xin Liang\IEEEauthorrefmark{6},
Guanpeng Li\IEEEauthorrefmark{2},
Franck Cappello\IEEEauthorrefmark{2}}
\IEEEauthorblockA{\IEEEauthorrefmark{1}University of Utah, Salt Lake City, UT, USA}
\IEEEauthorblockA{\IEEEauthorrefmark{2}Argonne National Laboratory, Lemont, IL, USA}
\IEEEauthorblockA{\IEEEauthorrefmark{3}University of California, Riverside, CA, USA}
\IEEEauthorblockA{\IEEEauthorrefmark{4}University of Iowa, Iowa City, IA, USA}
\IEEEauthorblockA{\IEEEauthorrefmark{5}Florida State University, Tallahassee, FL, USA}
\IEEEauthorblockA{\IEEEauthorrefmark{6}University of Kentucky, Lexington, KY, USA}

 tripti.agarwal@utah.edu, sdi1@anl.gov, jhuan380@ucr.edu, yafan-huang@uiowa.edu, ganesh@cs.utah.edu, \\ runderwood@anl.gov, kzhao@cs.fsu.edu, xliang@uky.edu, guanpeng-li@uiowa.edu, cappello@mcs.anl.gov

\thanks{Corresponding author: Sheng Di, Mathematics and Computer Science Division, Argonne National Laboratory, 9700 Cass Avenue, Lemont, IL 60439, USA}
}

\maketitle

\begin{abstract}

Error-bounded lossy compression has been a critical technique to significantly reduce the sheer amounts of simulation datasets for high-performance computing (HPC) scientific applications while effectively controlling the data distortion based on user-specified error bound. In many real-world use cases, users must perform computational operations on the compressed data (a.k.a. homomorphic compression). However, none of the existing error-bounded lossy compressors support the homomorphism, inevitably resulting in undesired decompression costs. In this paper, we propose a novel homomorphic error-bounded lossy compressor (called HoSZp), which supports not only error-bounding features but efficient computations (including negation, addition, multiplication, mean, variance, etc.) on the compressed data without the complete decompression step, which is the first attempt to the best of our knowledge. We develop several optimization strategies to maximize the overall compression ratio and execution performance. We evaluate HoSZp compared to other state-of-the-art lossy compressors based on multiple real-world scientific application datasets.   

\end{abstract}

\begin{IEEEkeywords}
Error-bounded Lossy Compression, Scientific Application, Homomorphism
\end{IEEEkeywords}

\section{Introduction}
\label{sec:intro}
Today's scientific applications tend to be running on extremely large execution scales, which may easily produce sheer amounts of simulation datasets that need to be kept in memory or stored in disks with limited storage capacity. Climate simulations, for example, may produce  200+ TB of data within 16 seconds \cite{cesm-bigdata-case}, and Fusion simulations can generate over
200 PB of data in a single run \cite{fusion-bigdata-case}. Such a large volume of simulation datasets may cause serious issues in data storage and transfer because of the limited storage space and data movement bandwidth (such as network, I/O, and memory).

Error-bounded lossy compression \cite{sz16,sz17,sz-auto,zfp,snapshot-compression-nbody,FAZ,SZx} has been proposed for years to resolve the above issues, especially because it can get fairly high compression ratios while strictly controlling the data distortion based on user-required error bound. For example, SZ and ZFP have been effective in significantly improving the I/O data writing performance, as shown in \cite{compression-io}. MDZ \cite{mdz} can be used to substantially reduce the storage size for Molecular Dynamics simulations while preserving the radio distribution function (RDF) very well. Wu et al. \cite{quantum-data-compression} developed an efficient lossy compression algorithm that can effectively compress the memory footprint for quantum computing simulations at runtime, which can significantly lower the requirement of memory capacity. Error-bounded lossy compression (e.g., Ocelot \cite{ocelet}) has also been used to improve the data transfer on a wide area network (WAN).  
Some general-purpose lossy compressors \cite{sperr, FAZ} can significantly reduce the scientific data size, though they may suffer relatively low compression speed. FAZ, for example, can compress large turbulence simulation data (Miranda \cite{miranda}) and seismic data (RTM \cite{RTM-tutorial}) by 93.6$\times$ and 514$\times$, respectively, at the relative error bound of $10^{-4}$ (a.k.a., 1E-4). Such compressors are very helpful in the use-case with very limited storage capacity or data transfer bandwidth. 

In addition to the above use cases which mainly make use of lossy compression to reduce storage size or mitigate data transfer cost, quite a few emerging use cases require performing certain operations on top of the compressed data. The existing compression methods, however, do not support performing various operations on the compressed data, so the users have to decompress the full dataset before executing the operations, inevitably introducing a high execution cost. For example, quantum circuit simulation \cite{quantum-data-compression}) may produce an extremely large amount of data to keep in memory, so it needs to compress the data to control the memory footprint. The data stored in the compressed format may need to be decompressed upon the need of simulation at runtime, which requires extra decompression steps inevitably for the traditional compressors. Another typical example is using lossy compression to reduce the communication cost to accelerate the overall MPI collective operation performance \cite{c-coll}. In the existing solutions, each process participating in the collaborative operation needs to fully decompress the compressed data received from another process, execute an arithmetic aggregation/reduction operation (such as addition), and then perform another compression on the aggregation dataset. If a compressor supports performing operations on the compressed data, the extra decompression cost can be saved or minimized, which can thus improve the overall performance in turn.   

In this paper, we develop a compression mechanism allowing to perform various arithmetic operations (such as negation, addition, and multiplication) on the error-bounded compressed datasets without expensive full decompression. We call such an operation-supported compression mechanism a \emph{homomorphic error-bounded lossy compression} in our paper.  \textit{Homomorphism} is defined as a structure-preserving mapping (e.g., the original data vs. compressed data), such that the arithmetic operations can be applied on top of compressed/encrypted data. 
As mentioned previously, homomorphic error-bounded lossy compression is very helpful in many emerging use cases such as reducing memory footprint and avoiding expensive decompression costs because of avoiding the full decompression step.  

Developing an efficient homomorphic error-bounded lossy compression is very challenging because of diverse compressor designs. In general, each existing error-bounded lossy compressor involves multiple steps from the data decorrelation to lossless encoding. For example, ZFP decorrelates the data by a blockwise near-orthogonal transform and SZ leverages various data prediction methods to do it. In order to reach high compression ratios, the lossy compressors often depend on sophisticated lossless encoders. For instance, ZFP adopts an embedded encoding and SZ chooses to use Huffman encoding + Zstd \cite{zstd}. Merging the arithmetic operations into these encoding techniques is quite non-trivial. 

Our proposed novel error-bounded lossy compression mechanism supports homomorphism, which is the first attempt in the error-bounded lossy compression community to the best of our knowledge. The fundamental idea is designing an efficient, lightweight compression pipeline that takes into account the execution of potential operations on the compressed data, and also minimizes the required steps and cost in the decompression in terms of performing the user-specified operations. We summarize the key contributions as follows: 
\begin{itemize}
    \item We develop an efficient error-bounded lossy compression method, which supports homomorphic compression, and also provide a theoretical proof for the guarantee of error control in our design. 
    \item We carefully optimize the homomorphism support based on various arithmetic operations (negation, addition, subtraction, multiplication, standard deviation, etc.) on compressed datasets without the full decompression. 
    \item We perform a comprehensive evaluation based on multiple real-world scientific datasets for our homomorphic compressor to show that our compressor is homomorphic and can improve the execution performance ranging from 1.30$\times$ to more than 245$\times$ when compared with traditional workflow for various homomorphic operations across different datasets. 
    \item We perform an experiment to show the effectiveness of HoSZp in the distributed environment using real-world data, with up to 2.08$\times$ performance speedup. Our code will be available publicly once the paper is accepted.
\end{itemize}

The remainder of the paper is organized as follows. Section \ref{sec:relatedwork} discusses the related works. We formulate the research problem in Section \ref{sec:problem}. We present the compression pipeline and homomorphic compression design in Section \ref{sec:Design}. We describe the key homomorphic operations integrated with the compression as well as the optimizations in Section \ref{sec:homomorphicOperations}. In Section \ref{sec:PerformaceEvaluation}, we provide the evaluation results as well as the analysis.  In the end, we conclude the paper with a vision of the future work in Section \ref{sec:conclude}. 
\section{Related Work}
\label{sec:relatedwork}

Error-bounded lossy compressors, such as ZFP \cite{zfp} and SZ\cite{sz16,sz17}, are often used for scientific data compression and can achieve a high compression ratio such as 50 or more, while strictly controlling the data distortion. 
None of the lossy compression techniques, however, were built with the goal of homomorphic operations on compressed data.  That is, if the users want to operate on data that has been compressed, they have to first fully decompress the data and then perform the operation. This will inevitably introduce undesired execution costs and memory overhead.

There exist some compression techniques that support homomorphic operations like Blaz \cite{martel2022compressed} and PyBlaz \cite{Tripti2023}. Blaz is a simple compressor that can only support 2-D arrays and can perform simple operations such as scalar addition, matrix addition, and multiplication of scalar. Since Blaz is a single-threaded sequential code, PyBlaz was built to support a more sophisticated compression setting. PyBlaz can support arbitrary dimensional data along with a lot more operations and measures. However, none of the above works provide a guarantee of compression error boundness for both the compression pipeline and operation built on top of them.

In our work, we mainly work with error-bounded lossy compressors, hence we provide a detailed literature survey of such compressors. These state-of-the-art error-bounded lossy compressors can be split into three models.

\textit{Prediction-based lossy compression model}. This compression model generally involves three key stages: data prediction, quantization, and lossless compression. The typical examples include SZ2 \cite{sz17}, SZ3 \cite{SZ3}, and FPZIP \cite{FPZIP}. SZ2, for example, adopts a hybrid prediction method combining Lorenzo predictor \cite{lorenzo} and linear regression. 

\textit{Transform-based lossy compression model}. The key idea of this compression model is performing data transform (such as wavelet transform) to convert the raw dataset to another coefficient dataset. ZFP \cite{zfp} and SPERR \cite{sperr} are two examples that uses this technique. This step can effectively decorrelate the raw dataset, such that a large majority of the data values are very close to 0. Then, an encoding method would be applied to significantly reduce the data size. ZFP, for example, applies a so-called big-plane-based embedded encoding method \cite{zfp} which stores only necessary bits with respect to user-required error bounds; SPERR performs a set partitioning embedded block coding algorithm (called SPECK \cite{speck}) to shrink the coefficient data size. 

\textit{Higher-order singular value decomposition (HOSVD) based compression model}. HOSVD \cite{ballester2019tthresh, tuckermpi} decomposes the data (i.e., a tensor) to a set of matrices and a small core tensor, with well-preserved $L_2$ normal error. By combining HOSVD and Tucker decomposition with other techniques such as bit-plane, run-length, and/or Core tensor arithmetic coding, the data size could be significantly reduced. 


To the best of our knowledge, none of the existing error-bounded lossy compressors support homomorphism. In this paper, we fill this gap by proposing a novel error-bounded lossy compressor namely, HoSZp that can perform certain homomorphic operations in a compressed data domain. 


\section{Problem Formulation}

\label{sec:problem}

We formulate the research problem as follows. Given a raw dataset (denoted by $D_r$), we denote the corresponding compressed data as $c$ and the decompressed dataset as $\hat{D_c}$. For a homomorphic error-bounded lossy compressor (such as \textit{HoSZp}), it would perform the user-required operation on the compressed data, leading to a new compressed data stream (denoted by $z$), whose corresponding decompressed dataset is denoted as $\hat{D_z}$.  
Basically, we focus on two types of operations: univariate operation that has one input dataset (denoted $f$($\cdot$), such as negation or adding a constant) and multi-variate operation each with two input datasets (denoted $g$($\cdot$,$\cdot$) such as adding up two vectors or data values). 

Figure \ref{fig:hocompression-flow} illustrates the entire workflow for the two types of operations. In the traditional workflow (see Figure \ref{fig:hocompression-flow} (a)), the user needs to get the decompressed dataset ($\hat{D}_c$) based on the compressed data stream ($c$) before operating $f(\cdot)$. In comparison, the homomorphic compression (i.e., marked as the \textit{new workflows}  in the figure) allows users to execute this operation on top of the compressed data format ($c$) without fully decompressing $c$. 

\begin{figure}[ht] 
\footnotesize
\centering
\subfigure[{Univariate-based Workflow}]
{
\raisebox{-1cm}{\includegraphics[scale=0.35]{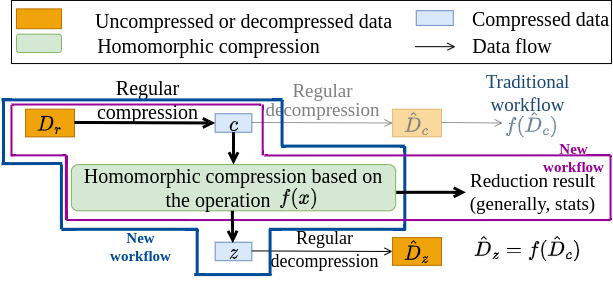}}%
}
\subfigure[{Multivariate-based Workflow}]
{
\raisebox{-1cm}{\includegraphics[scale=0.35]{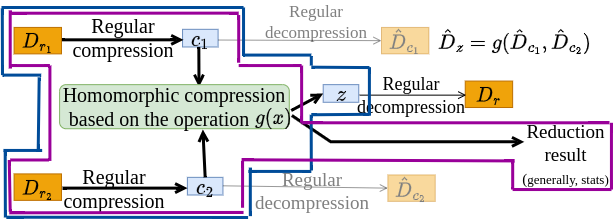}}%
}
\vspace{-1mm}
\caption{The Entire Workflow Regarding Homomorphic Compression}
\label{fig:hocompression-flow}
\end{figure}

Based on the two types of operations: univariate and multivariate, the output can also be split into two types as follows. 
\begin{itemize}
    \item \emph{Computation-as-output}: the output is a result based on a computation applied on the dataset (such as the mean, standard deviation or maximum value) -- illustrated as the purple workflow in Figure \ref{fig:hocompression-flow}. 
    \item \emph{Compression-as-output}: the output is another compressed data stream ($z$) with the specific operations already applied on top of the decompressed data -- illustrated as the blue workflow in Figure \ref{fig:hocompression-flow}.
\end{itemize}

For the compression-as-output workflow, the relationship between decompressed data for the univariate operation (denoted by $f(\cdot)$) and the decompressed data obtained from the homomorphic compressor must satisfy Formula (\ref{eq:uniho}).  
\def\formulaUniho{
\begin{array}{l}
 \hat{D}_z = f(\hat{D}_c)
\end{array}
}
\begin{equation}
\label{eq:uniho}
  \formulaUniho
\end{equation}

For the multivariate operation (denoted by $g$), the homomorphic error-bounded compressor must comply with Formula (\ref{eq:mulho}).
\def\formulaMulho{
\begin{array}{l}
 \hat{D}_z = g(\hat{D}_{c1}, \hat{D}_{c2})
\end{array}
}
\begin{equation}
\label{eq:mulho}
  \formulaMulho
\end{equation}
where $\hat{D}_{c1}$ and $\hat{D}_{c2}$ refer to the reconstructed datasets decompressed from two original raw datasets ($D_{r1}$ and $D_{r2}$).   

\begin{table}[ht]
    \footnotesize
    \caption{Notations used in HoSZp pipeline and Homomorphic operations.}
\centering
    
        \begin{tabular}{|c|c|}
        
            \hline
            \textbf{Notation} & \textbf{Description} \\
            \hline \hline
            $D_r$ & Raw Data \\
            \hline
            $c$ & Compressed Data \\
            \hline
            $\hat{D_c}$ & Decompressed Data \\
            \hline
            $z$ & Operated Compressed Data \\
            \hline
            $\hat{D_z}$ & Operated Decompressed Data \\
            \hline
            $g(\cdot)$ & multi-variate operation \\
            \hline
            $f(\cdot)$ & univariate operation \\
            \hline
            $\epsilon$ & user-defined error bound \\
            \hline
            $\mathbf{A_1}$ or $\mathbf{B_1}$ & block 1 of some $D_r$ (block\_size = \textit{m}'$\times$\textit{n}')\\
            \hline
            $\mathcal{O}_{\mathbf{A_1}}$ & Outlier of block $\mathbf{A_1}$ \\
            \hline$\varsigma_{\mathbf{A_1}}$ & Sign array of block $\mathbf{A_1}$ \\
            \hline
            $\varrho_{\mathbf{A_1}}$ & Quantized array for block $\mathbf{A_1}$ \\
            \hline
            $\mathcal{P}_\mathbf{A_1}$ & Predicted array of block $\mathbf{A_1}$ \\
            \hline
            $\mathcal{C}_{\mathbf{A_1}}$ & Bits for compressed block $\mathbf{A_1}$ \\
            \hline
        \end{tabular}

    \label{tab:notations}
\end{table}
The key objective of the research is to develop an efficient homomorphic error-bounded compressor, which can avoid the expensive full data decompression when performing various operations on top of the compressed data. It is worth noting that completely avoiding decoding during the homomorphic operation is impossible. Instead, our design motivation/objective is to minimize the decoding work as much as possible by keeping only necessary steps concerning homomorphic operations.  

Table \ref{tab:notations} summarizes all notations used in the paper, which helps discuss the design of our designed HoSZp and homomorphic operations explained in Section \ref{sec:Design} and Section \ref{sec:homomorphicOperations}.
\section{Design Overview}
\label{sec:overview}

\label{sec:Design}



The HoSZp mainly aims to support homomorphism while still respecting error-bound features and expecting to reach relatively high compression ratios and high compression/decompression performance. Towards this end, we substantially improve the cuSZp compression pipeline \cite{cuszp} (which was initially designed for only GPU) by developing a new multi-threaded CPU version (we call it \textit{SZp}) and enabling it to support homomorphism. Compared with classic SZ's original design \cite{sz16,sz17}, our compression pipeline features higher compression and decompression speed, also being much more suitable for homomorphic compression, although with gracefully degraded compression ratios.

\subsection{Compression Pipeline}
\label{sec:pipeline}
HoSZp is a floating-point data compressor consisting of three main steps: Quantization (QZ), Decorrelation (LZ), and Blockwise Fixed length byte Encoding (BF), as shown in Figure \ref{fig:pipeline}. 
\begin{figure} [ht]
\footnotesize
    \centering
    \includegraphics[width=0.48 \textwidth]{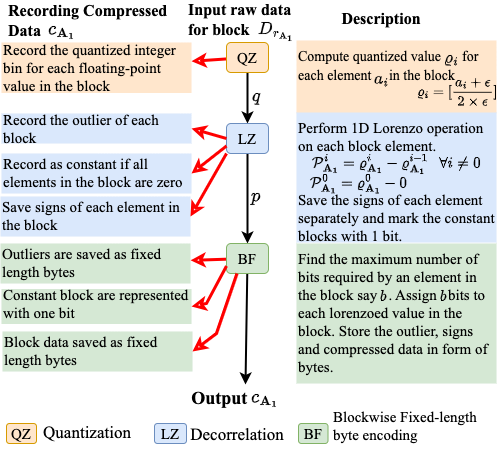}
    \vspace{-2mm}
    \caption{HoSZp compression pipeline (workflow). Decompression of HoSZp is the inverse of all the steps.}
    \label{fig:pipeline}
\end{figure}


In what follows, we describe the HoSZp pipeline using 2D arrays, $D_{r_1}$ and $D_{r_2}$, which can be extended to other dimensions (e.g., 1D and 3D) easily. Without loss of generality, each of the two 2D arrays consists of $m \times n$ elements/data points. We split each dataset into a certain block, and the block size is set to $m' \times n'$, which results in a total of $\frac{m}{m'} \times \frac{n}{n'}$ blocks for both  $D_{r_1}$ and $D_{r_2}$. These blocks are denoted as $D^1_{r_1}, D^2_{r_1}, \dots, D^{\frac{m}{m'} \times \frac{n}{n'}}_{r_1}$, and  $D^1_{r_2}, D^2_{r_2}, \dots, D^{\frac{m}{m'} \times \frac{n}{n'}}_{r_2}$,  for arrays $D_{r_1}$ and $D_{r_2}$, respectively which can be compressed independently. For simplicity, without loss of generality, we describe our homomorphic compression and the pipeline using one specific block ($D^1_{r_1}$, and $D^1_{r_2}$) from the two datasets, respectively (as shown in Figure \ref{fig:pipeline}). We denote the elements in $D^1_{r_1}$ as $\mathbf{A_1} = a^i_1$, and those in $D^1_{r_2}$ as an array $\mathbf{B_1} = b^i_1$, where $i = m' \times n'$ is the number of elements in the block and set the user-defined absolute error bound to $\epsilon$. 

We now explain each step of HoSZp using block $\mathbf{A_1}$, with a few more new defined notations wherever necessary. 

\textbf{1. Quantization (QZ):}
This step converts the entire dataset into integers based on the user-defined error bound ($\epsilon$). The quantized value (a.k.a., quantization bin number) is given by the Formula (\ref{eq:quantizationEq}).
\def\formulaMulho{
\begin{array}{l}
  \varrho^i_\mathbf{A_1} = \Big\lfloor\frac{a^i_1+\epsilon}{2 \times \epsilon}\Big\rfloor
\end{array}
}
\vspace{-2mm}
\begin{equation}
\label{eq:quantizationEq}
  \formulaMulho
  \vspace{-2mm}
\end{equation}
, where $a^i_1$ is the $i$-th floating-point value in block $\mathbf{A_1}$ and $\lfloor \rfloor$ is a floor function. This step converts all the floating-point data into integer numbers (i.e., quantization bin numbers) because the transformed data are easier to process by lossless encoders such as Huffman encoding. Note that the data reproduced based on the quantization bins during the inversion of this step are not the same as the original data, and the data loss is limited within the error bound.

\textbf{2. Decorrelation (LZ):}
In this step, we exploit that most scientific data are spatially adjacent-correlated (meaning that the data close to each other or within a region have similar value ranges). Hence, we apply a 1-D Lorenzo operator \cite{https://doi.org/10.1111/1467-8659.00681} on each block given in Formula (\ref{eq:predictionEq}). This helps in further decorrelating the integer values to reduce the necessary bits to store. We further store the \emph{outlier} (\textbf{first value of each block}) separately, represented as $\mathcal{O}_\mathbf{A_1}$. This information helps create homomorphic operations, which will be detailed later.

\begin{equation}
\label{eq:predictionEq}
  \mathcal{P}^i_{\mathbf{A_1}} = \varrho^i_\mathbf{A_1} - (\varrho^{i-1}_\mathbf{A_1} \text{ if } i \neq 0; 0 \text{ otherwise})
\end{equation}

We use an example to explain decorrelation further. Suppose the block size for a 2-D input ($D_r$) is $2 \times 2$ and we have the following data values in the block $\mathbf{A_1}$ = \{-0.025, -0.025, -0.051, -0.052\}. Then, the quantized integers for each value are $\varrho_\mathbf{A_1} = \{-1, -1, -3, -3\}$. A 1D Lorenzo predictor (subtracting each value from its corresponding previous neighbor) is applied on the on $\varrho_\mathbf{A_1}$, resulting in predicted values $\mathcal{P}_{\mathbf{A_1}} = \{0, 0, -2, 0\}$ and the outlier $\mathcal{O}_{\mathbf{A_1}} = -1$.


We store the sign of each element separately, which removes the ambiguity that can occur during the succeeding lossless encoding step, i.e., fixed-length encoding. Positive numbers are represented with a bit value of $0$, and negative are represented with a bit value of $1$. Hence, $\mathcal{P}_{\mathbf{A_1}}$=$\{0, 0, 2, 0\}$, $\mathcal{O}_{\mathbf{A_1}}$=$-$1, and $\varsigma_{\mathbf{A_1}} = \{0, 0, 1, 0\}$, where  $\varsigma_{\mathbf{A_1}}$ is the sign array for block $\mathbf{A_1}$.


\textbf{3. Blockwise Fixed-length byte encoding (BF):} The data obtained from the prediction $\mathcal{P}_{\mathbf{A_1}}$ is converted into bits using a fixed-length encoding technique. In this method, the maximum number of bits required for an integer in a given block $\mathbf{A_1}$ is calculated, and then all the numbers are stored with the same number of bits, represented by $\mathcal{C}_\mathbf{A_1}$. This reduces the number of bits (converted to bytes) required to store all the elements.

In the above example, for predicted array $\mathcal{P}_{\mathbf{A_1}} = \{0, 0, 2, 0\}$, the maximum number of bits taken by the block is $2$ bits by integer $2$, hence the entire block can be represented with $8$ bits, i.e $\mathcal{C}_\mathbf{A_1} = (00001000)_2 = (8)_{16}$.  Note that if all the elements in a $\mathcal{P}_{\mathbf{A_1}}$ (except for the outlier $\mathcal{O}_{\mathbf{A_1}}$) have integer value $0$ for any block $\mathbf{A_i}$, we indicate such a block as a constant block and represent it with bit $0$.
\begin{figure} [ht]
\vspace{-4mm}
\footnotesize
    \centering
    \includegraphics[width=0.48 \textwidth]{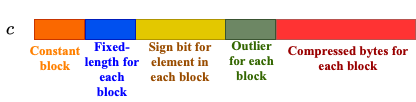}
    \vspace{-3mm}
\caption{Representation of compressed data}    
    \label{fig:compressed_data}
\end{figure}

Finally, the compressed data is stored as follows: Fixed-length for each block followed by outlier for each block $\mathcal{O}_{\mathbf{A_j}}$, followed by sign bits for each element of all blocks $\varsigma^i_{\mathbf{A_j}}$ and then the compressed bits of each block $\mathcal{C}_\mathbf{A_j}$, where is $j = \Huge\{1, 2, \ldots,\{\frac{m}{m'} \times \frac{n}{n'}\}\Huge\}$. Figure \ref{fig:compressed_data} shows a simple compressed data representation.

\subsection{Homomorphic Compression}
\label{sec:HomoComp}
The compression and decompression pipelines of HoSZp are designed toward homomorphic operations, as shown in Figure \ref{fig:operationWorkflow}. The \textbf{traditional workflow operation} performs the full decompression (i.e., decompress the fixed-length encoded bytes for each block, then the inverse of Lorenzo operation, and finally the inverse of the quantization step). The desired operation is then applied to the decompressed data, and full compression, including quantization, Lorenzo, and blockwise fixed-length encoding, is again applied to the operated data to obtain the compressed format. In the \textbf{Homomorphic operation workflow}, the main idea is avoiding the full decompression and full compression on the operated data as discussed above in the traditional workflow. This involves skipping the steps of decompression and corresponding compression as required so that the operated compressed data is homomorphic to the operated compressed data obtained using traditional workflow. 

\begin{figure}[ht] 
\footnotesize
\centering
\subfigure[{Traditional Workflow}]
{
\raisebox{-1cm}{\includegraphics[scale=0.39]{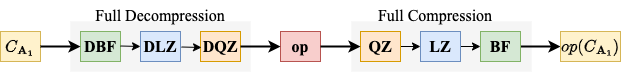}}%
}

\subfigure[{Homomorphic operation Workflow}]
{
\raisebox{-1cm}{\includegraphics[scale=0.41]{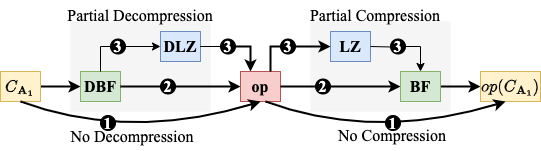}}%
}
\vspace{-1mm}
\subfigure
{
\raisebox{-1cm}{\includegraphics[scale=0.39]{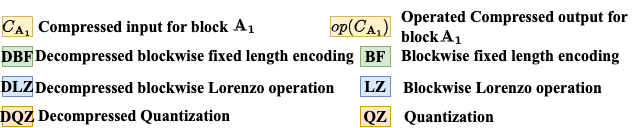}}%
}
\vspace{-1mm}
\caption{Illustration of Homomorphic compression vs. tradition workflow}
\label{fig:operationWorkflow}
\end{figure} 

Depending on the datasets and operations, there are three different ways of performing the operations. \circled{1} \textit{Directly performing operations on the input compressed data.} Operations like negation and addition of scalar on compressed data can be performed in a fully compressed space because the compressed data consists of signs and outliers saved separately, which can be used to calculate these operations. \circled{2} \textit{Performing operations on the decompressed blockwise fixed-length byte data.} Operations such as the element-wise addition of two compressed data use the decompressed blockwise fixed-length byte data to operate and then compress the data back by performing the blockwise fixed-length encoding. \circled{3} \textit{Performing operations by first decompressing data using an inverse of blockwise fixed length encoding and inverting the Lorenzo operation.} Operations such as multiplication of scalar, mean, variance, standard deviation, hardamard product, and covariance use this workflow, and the operated data obtained is again compressed back by applying the Lorenzo operator and then performing blockwise fixed-length encoding. Details of how each of the operations is performed are explained in Section~\ref{sec:homomorphicOperations}.

\subsection{Performance and Quality Analysis}
\label{sec:Performanceand QualityAnalysis}

In this section, we analyze the substantial advantage of our homomorphic compression design. 

\subsubsection{Analysis of Memory Cost and Performance for HoSZp}

The first advantage of HoSZp over the traditional workflow (SZp) is the former does not need the full decompression operation, which can save memory costs a lot at runtime. In addition, HoSZp can also significantly improve the performance of the overall execution with guaranteed correct/identical operation results, which are analyzed as follows.

Using the pipeline discussed above, the performance of our homomorphic operations could be much higher than that of traditional operations. The conventional execution requires fully decompressing the data before performing operations and then compressing the newly generated data. In comparison, the homomorphic operations (details in Section ~\ref{sec:homomorphicOperations}) do not require fully decompressing the data but need to operate on partially decompressed data in the intermediate stage, which would take less time as compared to full decompression. More specifically, the overall operation cost can be significantly reduced, especially when the error bound is relatively high. This is because the QZ+LZ potentially leads to quite a few all-zero blocks (i.e., the blocks containing only zero-value quantization bins), for which the numerical operations (such as addition, subtraction, negation) can be applied on the blockwise metadata (i.e., outliers) instead of each data value.


\subsubsection{Analysis of Homomorphism for HoSZp}
The operations designed in HoSZp are homomorphic to the operations performed on decompressed data. We provide a strict proof below.

\begin{theorem}
\label{thm:1}
The dataset ($\hat{D_z}$) reconstructed from the HoSZp-compressed bytes ($z$) based on an univariate operation $f(\cdot)$ or multivariate operation $g(\cdot)$ is identical to the results of applying $f(\cdot)$ or $g(\cdot)$ on the fully-decompressed datasets $\hat{D_c}$.
\end{theorem}

\begin{proof}
The intuitive proof idea is that none of the operations are performed by reverting the quantization (\textbf{Q}) step, which is the only lossy step in the entire pipeline. Therefore, all the operations are homomorphic to the operations when performed by the traditional SZp pipeline. 

In what follows, we provide a detailed proof for the scalar addition operation (as an example), and the correctness with other operations can be proved similarly. 
\\
\textbf{Given:}
 Predicted values of block \(\mathbf{A_1}\) are \(\mathcal{P}^1_{\mathbf{A_1}}, \mathcal{P}^2_{\mathbf{A_1}}, \ldots, \mathcal{P}^{m' \times n'}_{\mathbf{A_1}}\), Outlier is \(\mathcal{O}_\mathbf{A_1}\), Scalar value \( s \) and \( \varrho_s = \frac{s}{2\epsilon} \) where $\epsilon$ is the error-bound.
 
\textbf{Process:}
\begin{enumerate}
    \item \textit{Traditional Scalar Addition:}\\
    \textit{Step 1:} Reverse decorrelation. \((\mathcal{O}_\mathbf{A_1} + \sum_{i=1}^k \mathcal{P}^i_{\mathbf{A_1}})\) for \( k = 1, 2, \ldots, m' \times n' \).\\
    \textit{Step 2:} Reverse quantization. 2$\epsilon$\((\mathcal{O}_\mathbf{A_1} + \sum_{i=1}^k \mathcal{P}^i_{\mathbf{A_1}})\).\\ \textit{Step 3:} Add scalar \( s \). 2$\epsilon$\((\mathcal{O}_\mathbf{A_1} + \sum_{i=1}^k \mathcal{P}^i_{\mathbf{A_1}})\) + s.
    
    \item \textit{Homomorphic Scalar Addition:}\\
    \textit{Step 1:} Add \( \varrho_s \) to \( \mathcal{O}_\mathbf{A_1} \), resulting in \( \mathcal{O}'_\mathbf{A_1} = \mathcal{O}_\mathbf{A_1} + \varrho_s \).\\
        \textit{Step 2:} Reverse decorrelation. \((\mathcal{O}'_\mathbf{A_1} + \sum_{i=1}^k \mathcal{P}^i_{\mathbf{A_1}})\).\\
        \textit{Step 3:} Reverse quantization. 2$\epsilon$\((\mathcal{O}'_\mathbf{A_1} + \sum_{i=1}^k \mathcal{P}^i_{\mathbf{A_1}})\)
    
\end{enumerate}

\textbf{To Prove:}
The decompressed data from processes (1) and (2), obtained after Step 3, are equivalent.

\textbf{Proof:}
For process (2), substituting \( \mathcal{O}'_\mathbf{A_1} = \mathcal{O}_\mathbf{A_1} + \varrho_s \) in $2\epsilon(\mathcal{O}'_\mathbf{A_1} + \sum_{i=1}^k \mathcal{P}^i_{\mathbf{A_1}})
\vspace{-3mm}= 2\epsilon\mathcal{O}_{\mathbf{A_1}} + 2\epsilon\varrho_s + 2\epsilon\sum_{i=1}^k\mathcal{P}^i_{\mathbf{A_1}}$
\begin{multline*}
= 2\epsilon\mathcal{O}_{\mathbf{A_1}} + 2\epsilon\left(\frac{s}{2\epsilon}\right) + 2\epsilon\sum_{i=1}^k\mathcal{P}^i_{\mathbf{A_1}} 
= 2\epsilon\mathcal{O}_{\mathbf{A_1}} + s + 2\epsilon\sum_{i=1}^k\mathcal{P}^i_{\mathbf{A_1}}
\end{multline*}


Hence, the decompressed data obtained from the traditional and homomorphic scalar addition processes are identical, thus proving the equivalence of the two approaches.
\end{proof}

We also validate the correctness of the above theorem using experiments with real-world datasets in Section \ref{sec:DataVisualization}.

\section{Homomorphic Operations for HoSZp}
\label{sec:homomorphicOperations}
In this section, we discuss the different homomorphic operations we developed in HoSZp. In the following, we still mainly describe our design based on the blocks ($\mathbf{A_1}$ and $\mathbf{B_1}$) from the two compressed datasets without loss of generality. After applying quantization and prediction, we obtain specific metadata: the outlier for blocks $\mathbf{A_1}$ and $\mathbf{B_1}$, denoted as $\mathcal{O}_\mathbf{A_1}$ and $\mathcal{O}_\mathbf{B_1}$, the predicted values that are represented as arrays $\mathcal{P}_\mathbf{A_1}$ and $\mathcal{P}_\mathbf{B_1}$, and the sign elements represented as arrays $\varsigma_\mathbf{A_1}$ and $\varsigma_\mathbf{B_1}$, respectively. We also use intermediate quantized values for some of the operations and denote the quantized values as $\varrho_\mathbf{A_1}$ and $\varrho_\mathbf{B_1}$, respectively. Note that these quantized and predicted values are integers, which are subsequently stored as bytes using a fixed-length byte encoding scheme.

Using the above notations (also summarized in Table \ref{tab:notations}), we explain different homomorphic operations supported by HoSZp (listed in Table \ref{tab:operations}) along with examples wherever necessary. Some operations are derivable from other operations; hence, we discuss those operations briefly.
\begin{table}
    \centering
    \footnotesize
    \caption{List of operations in HoSZp, along with the type of operation and the result type obtained after the operation is applied. Note that, all the operations follow the error-boundness. This is because none of the operations apply inverse quantization on the input compressed data.}
    \vspace{-3mm}
\resizebox{0.99\columnwidth}{!}{  
    \begin{tabular}{|c|c|c|}
    \hline
        \textbf{Type} &  \textbf{Operation} & \textbf{Result Type}\\
    \hline \hline
        Univariate Operation & Negation & \textcolor{blue}{Compression-as-output}\\
        Univariate Operation &  Scalar addition & \textcolor{blue}{Compression-as-output}\\
        Univariate Operation & Scalar subtraction & \textcolor{blue}{Compression-as-output}\\
        Univariate Operation & Scalar multiplication & \textcolor{blue}{Compression-as-output}\\
        \hline
        Univariate Reduction& Mean & \textcolor{red}{Computation-as-output}\\
        Univariate Reduction & Variance & \textcolor{red}{Computation-as-output}\\
        Univariate Reduction & Standard Deviation & \textcolor{red}{Computation-as-output}\\
        \hline
        Bivariate Operation & Element-wise addition & \textcolor{blue}{Compression-as-output}\\
        Bivariate Operation & Element-wise subtraction & \textcolor{blue}{Compression-as-output}\\
        Bivariate Operation & Hardamard Product & \textcolor{blue}{Compression-as-output}\\
        \hline
        Bivariate Reduction & Covariance & \textcolor{red}{Computation-as-output}\\
         
        \hline
    \end{tabular}}
    
    \label{tab:operations}
\end{table}
\subsection{Scalar Homomorphic Operations}
\label{sec:scalarop}
Scalar Homomorphic Operations are point-wise operations that are performed on each element of the compressed dataset (or matrix). We describe each scalar homomorphic operation available in HoSZp here in detail.

\subsubsection{Negation}
Negation operation \cite{bronson1989theory} is a unary operation and is solely dependent on reversing the signs of the data (saved explicitly in our compressed data), making the operation in fully compressed space\footnote{Fully compressed space means that the compressed bits saved in our compressed data are not even partially decompressed}. 
Consider a single block array $\mathbf{A_1}$, then
the negation operation is performed as follows: Invert the signs of each element in the array $\varsigma_\mathbf{A_1}$ to obtain the inverted signs, $\neg \varsigma_{\mathbf{A_1}}$. This is done by applying a logical NOT operation ($\neg$) element-wise: $\neg \varsigma_{\mathbf{A_1}} = \{\neg \varsigma_{a^1_0}, \neg \varsigma_{a^1_1}, ..., \neg \varsigma_{a^1_{\{m' \times n'\} - 1}}\}$.


    




\subsubsection{Scalar Addition}
The scalar addition \cite{bronson1989theory} involves adding a constant scalar value to an input array. In our compressor, this is done by calculating the quantized bin index of the scalar $s$ based on the user-defined error (let the quantized bin index be $\varrho_s$) and then by adding the scalar value to the outliers $\mathcal{O}$ of each block. Since we save the $\mathcal{O}$ separately, this operation is also performed in a fully-compressed space.

Suppose we want to add a value say $0.67$ to $\mathbf{A_1}$. The quantized bin value for $s = 0.67$ will be $\varrho_s = 33$. Hence adding $\varrho_s$ to the outlier of $\mathbf{A_1}$ i.e. $\mathcal{O}_{\mathbf{A_1}} + \varrho_s = -1 + 33 = 32$. Finally, the metadata for the scalar addition will result in $\mathcal{O}_{\mathbf{A_1}} = 32$,  $\mathcal{P}_{\mathbf{A_1}} = \{0, 0, 1, 0\}$ and $\varsigma_{\mathbf{A_1}} = \{1, 1, 0, 1\}$.


    

\subsubsection{Scalar Subtraction}
Scalar subtraction \cite{bronson1989theory} involves subtracting a scalar value ($s$) from the matrix. This is similar to scalar addition, but here the scalar quantized value ($\varrho_s$) is deducted from the outliers $\mathcal{O}$ of each block. This operation is also performed in full compressed space.
\subsubsection{Scalar Multiplication}
Scalar multiplication  \cite{bronson1989theory} involves multiplying an element in a matrix. Since the values in the matrix obtained are predicted values and the prediction is made based on addition operations, it is impossible to perform multiplication without exact quantized values for each block. Hence, we revert the matrix for multiplication to obtain the corresponding quantized values denoted as $\varrho_\mathbf{A_1}$. We then get the quantized value of the scalar $s$ as $\varrho_s$ and multiply it with $\varrho_\mathbf{A_1}$. These are then reversed into compressed form to obtain a compressed scalar multiplied matrix. As we do have to decompress the data for scalar multiplication partially, this operation is performed in partially decompressed space\footnote{Partially decompressed space is defined as space, where the entire decompression pipeline is not performed instead some steps of decompression, are performed to obtain the desired results.}.

Suppose we want to multiply a scalar value $s = 3.14$ by $A_1$. The quantized bin value for $s$ will be $\varrho_s = 157$. Multiplying $\varrho_s$ to the quantized values for the block $\varrho_\mathbf{A_1} = \{-1, -1, -3, -3\}$ results in $\varrho_\mathbf{A_1} = \{-157, -157, -471, -471\}$ which is then divided by error produced by quantization of scalar value ($2 \times \epsilon$). The metadata for this scalar multiplication will be $\varrho_\mathbf{A_1} = \{-3, -3, -9, -9\}$. Finally, the compressed data will have 
$\mathcal{O}_\mathbf{A_1} = -3$,  $\mathcal{P}_\mathbf{A_1} = \{0, 0, 6, 0\}$ and $\varsigma_\mathbf{A_1} = \{0, 0, 1, 0\}$.


    

    
\subsection{Univariate Homomorphic Reductions}
\label{sec:scalarReduction}
Univariate homomorphic reductions are the operations performed on one compressed data (or matrix), which results in a single floating-point value. We describe each homomorphic reduction available in HoSZp here in detail.
\subsubsection{Mean} Mean \cite{underhill1996introstat} is calculated as the sum of all the elements in the matrix divided by the total number of elements. The quantized values of the block ($\varrho_\mathbf{A_1}$) are summed together to get block-wise addition. These block-wise additions are then divided by the total number of elements ($m \times n$) to obtain the mean of the entire matrix.  This process produces a final decompressed mean value instead of compressed data. Note that the same kernel can be used to calculate block-wise means by adding the quantized elements $\varrho_\mathbf{A_i}$ where $i = \frac{m}{m'} \times \frac{n}{n'}$ of each block and dividing each block by the number of elements in the block ($m' \times n'$).

Suppose, we want to find the mean of $\mathbf{A_1}$ where the quantized values  $\varrho_\mathbf{A_1} = \{-1, -1, -3, -3\}$. These values are then added together, resulting in $-8$, which is then divided by the number of elements (here $m' \times n' = 4$) and then finally multiplied by $2 \times \epsilon$ to get the final mean value of $-0.04$.
\subsubsection{Variance} Variance \cite{wasserman2004all} is similar to the mean operation. Still, each quantized value is first subtracted from the mean of the matrix, and then the obtained value is squared and added to get block-wise additions. The block-wise additions are then summed together and divided by the total number of elements ($m \times n$) to obtain the variance of the entire matrix. 
\subsubsection{Standard Deviation} Standard deviation \cite{bland1996measurement} operation is similar to variance operation, which is calculated by taking the square root of the variance of the compressed data.
\subsection{Bivariate Homomorphic Operations}
\label{sec:bivariateop}
Bivariate operations are performed between two compressed data (or matrices) and may result in the same or different numbers of elements for each operation. We describe each bivariate homomorphic operation supported in HoSZp as follows.
\subsubsection{Element-wise Addition}
Element-wise addition \cite{bronson1989theory} involves adding each element in the two input compressed data (of the same dimensions) according to their index value. In our compression technique, this can be achieved by adding the outlier ($\mathcal{O}_\mathbf{A}$) of each block in the first compressed data to the corresponding outlier ($\mathcal{O}_\mathbf{B}$) of each block in the other compressed data, respectively. We also need to add the signed predicted values ($\mathcal{P}_\mathbf{A}$ and $\mathcal{P}_\mathbf{B}$) for each block to obtain $\mathcal{P}_\mathbf{A+B}$

Let us add $\mathbf{A_1}$ and $\mathbf{B_1}$, then the outliers of $\mathbf{A_1}$ and $\mathbf{B_1}$ i.e. $\mathcal{O}_\mathbf{A_1}$ and $\mathcal{O}_\mathbf{B_1}$ are added along with their corresponding predicted indices i.e $\mathcal{P}_\mathbf{A_1}$ and $\mathcal{P}_\mathbf{B_1}$. In general, $\mathcal{O}_\mathbf{A_1}$ and $\mathcal{O}_\mathbf{B_1}$ and $\mathcal{P}_\mathbf{A_1}$ and $\mathcal{P}_\mathbf{B_1}$ are added together. Note that here $\mathcal{P}_\mathbf{A_1}$ and $\mathcal{P}_\mathbf{B_1}$ are predicted signed values.
In the above example, adding $\mathbf{A_1}$ and $\mathbf{B_1}$ results in following metadata $\mathcal{O}_\mathbf{A_1 + B_1} = 2$, $\mathcal{P}_\mathbf{A_1 +B_1} = \{0, 0, 2, 0\}$ and the updated sign elements will be $\varsigma_\mathbf{A_1 + B_1} = \{0, 0, 0, 0\}$.


    

\subsubsection{Element-wise Subtraction} Element-wise subtraction \cite{bronson1989theory} involves subtracting each element in the two input compressed data according to their index value.  This is similar to element-wise addition, but these are subtracted instead of adding the elements. Note that subtraction is a non-commutative operation; hence subtracting $A - B$ will not result in the same values as $B-A$.
\subsubsection{Hardamard Product} Hardamard product \cite{horn2012matrix, million2007hadamard} is the element-wise multiplication of each element in the two input compressed data according to their index value. This is similar to scalar multiplication, but instead of multiplying a scalar quantized value, the quantized values of each matrix are multiplied according to their indices.
\subsection{Bivarate Homomorphic Reductions}
\label{sec:BivariateReduction}
Bivariate Homomorphic Operations are similar to univariate homomorphic operations but here the operations are performed on two compressed data (or matrices) and result in a single floating-point value. We briefly describe each bivariate homomorphic reduction available in HoSZp here.
\subsubsection{Covaraince} Covariance \cite{park2018fundamentals} is similar to how the variance of a single compressed data is obtained. Instead, we perform a similar procedure for two compressed data to get the covariance.
\subsubsection{SSIM} SSIM  \cite{SSIM} is a structural similarity index measure that helps determine the image quality that degraded due to data compression. This operation involves the mean, variance, and covariance of the input matrix. The operation can be easily deducible using homomorphic operations discussed before; hence, we do not elaborate on it in the paper.

\begin{table}[ht]
\vspace{-3mm}
\footnotesize
    \centering    
    \caption{Scientific simulation real-world data used in the evaluation.}
    \vspace{-2mm}
    \begin{tabular}{|c|c|c|c|}
    \hline
        \textbf{Datasets} &  \textbf{\# of fields}
        & \textbf{ Dimension} &  \textbf{Data size}\\
        \hline    \hline        
        \textbf{Hurricane} & 7 & $500 \times 500 \times 100$ & 1.25GB \\
         \hline   
        \textbf{CESM-ATM} & 5 & $3600 \times 1800$ & 1.47GB \\
         \hline   
        \textbf{SCALE-LETKF} & 12 & $98 \times 1200  \times 1200$ & 4.9GB\\
         \hline   
        \textbf{Miranda} & 7 & $256 \times 384 \times 384$ & 1.87GB\\
         \hline
    \end{tabular}
    
    \label{tab:datasets}
    \vspace{-3mm}
\end{table}
\section{Performance Evaluation and Analysis}
\label{sec:PerformaceEvaluation}
In this section, we evaluate the performance of HoSZp for different operations using 4  scientific applications  across different domains (Section \ref{sec:ExperimentalSetup}). We evaluate time performance breakdown (Section \ref{sec:TimeCalculation}), throughput (Section \ref{sec:ThroughputCalculation}),  compression ratio (Section \ref{sec:compressionRatio}), and data visualization (Section \ref{sec:DataVisualization}) based on these datasets with different error bounds. We show the significant performance improvement of HoSZp operations over the traditional compression+operation workflow operated based on SZp, which is an outstanding ultra-fast error-bounded lossy compressor \cite{cuszp}. 

\subsection{Experimental Setup}
\label{sec:ExperimentalSetup}
\subsubsection{Platforms}
All the experiments are performed on a the LCLS Bebop supercomputer. Each node has two Intel Xeon E5-2695 v4 processors and a 128 GB of DRAM.  Each node on Bebop consists of 36 cores, and our multi-threaded code uses all 36 cores per node.
\subsubsection{Datasets}
The datasets used for experiments are four varied types of floating-point scientific data, as listed in Table ~\ref{tab:datasets}. These datasets are taken from Scientific Data Reduction Benchmarks \cite{sdrb} from various domains, i.e., weather simulation (Hurricane ISABEL \cite{hurricane}), climate simulation (CESM-ATM) \cite{cesm-code}, climate simulation (SCALE-LETKF) \cite{scale-letkf}, and turbulence simulation Data (Miranda \cite{miranda}).  They are commonly used to evaluate different lossy compressors available in various works of literature \cite{sz-auto,SZ3,cuszp}. 

\subsubsection{Evaluation Metrics} For evaluating the homomorphic operations provided in HoSZp, we perform time cost analysis, throughput analysis,  compression ratio, and data reconstruction quality. Below are the details of these evaluation metrics.

\begin{itemize}
    \item \textbf{Time Cost (in seconds)} helps determine the runtime a compressor takes to perform a compression or decompression in a compressor. In HoSZp, we measure the time cost by one or more kernels for its execution; hence, the total time is the sum of the time cost by each kernel execution. For a traditional workflow of SZp, the time cost to operate is the sum of decompression time, the time taken to operate, and the compression time to compress the operated data.
    \item \textbf{Throughput (GB/s)} helps determine the gigabytes of data that a compressor can process in the entire process: total data divided by the time cost to process that data.  
    \item \textbf{Compression Ratio} is the ratio of original data size to the compressed data size. 
    We will show that our HoSZp has even higher compression ratios than SZp. This is because there is no extra storage overhead in our homomophic compression design, and our design can compress the blocks with outliers more effectively.
    \item \textbf{Visualization} is important to understanding data reconstruction quality. To this end, we decompress HoSZp-generated homomorphic data and compare it with the results by applying the same operation on the traditionally decompressed data (i.e., decompress + operation).
\end{itemize}

\subsection{Performance Evaluation}

First of all, we evaluate the overall performance of the traditional compression+operation+decompression workflow based on multiple state-of-the-art error-bounded lossy compressors (including SZp, SZ2, SZ3, SZx and ZFP). As shown in Table \ref{tab:performanceComparisionWithOthers}, SZp significantly outperforms all other compressors (about 1.5$\times$ speedups over the second-best one -- SZx). The key reason is that SZp has the highest throughput in both compression and decompression from among all the compressors here, and the compression/decompression cost is the major bottleneck of the whole workflow, despite lower compression ratios compared with other compressors (to be shown later). Since SZp is the best compressor for the traditional workflow, we mainly compare our HoSZp with SZp in the following text, without loss of generality. 
\begin{table}[ht]
\vspace{-2mm}
    \centering  
    \caption{Throughput (MB/sec) for different operations on Hurricane Dataset using multiple compressors with $\epsilon$=1E-4. This experiment is performed by first performing compression on the data, then decompressing the data, and finally applying different operations on the decompressed data.}
    \vspace{-1mm}
\resizebox{0.98\columnwidth}{!}{%
    \begin{tabular}{|c|c|c|c|c|c|}
    \hline
        \textbf{Operations} 
        & \textbf{SZp} &  \textbf{SZ2} & \textbf{SZ3} & \textbf{SZx} & \textbf{ZFP} \\
        \hline \hline       
        \textbf{Negation} & 384 & 100 & 81 & 264 & 108\\
        \hline 
        \textbf{Scalar addition} & 358 & 99 & 80 & 251 & 105\\
        \hline 
        \textbf{Scalar subtraction} & 369 & 99 & 81 & 257 & 106\\
        \hline 
        \textbf{Scalar multiplication} & 366 & 99 & 81 & 255 & 106\\
        \hline 
        \textbf{Mean} & 381 & 100 & 81 & 262 & 107\\
        \hline 
        \textbf{Variance} & 287 & 92 & 76 & 214 & 98\\
        \hline 
        \textbf{Standard Deviation} & 294 & 93 & 77 & 218 & 99\\
        \hline 
        \textbf{Element-wise addition} & 354 & 98 & 80 & 249 & 105\\
        \hline 
        \textbf{Element-wise subtraction} & 354  & 98 & 80 & 249 & 105  \\
        \hline 
        \textbf{Element-wise multiplication} & 357 & 99 & 80 & 251 & 106\\
        \hline 
        \textbf{Covariance} & 230 & 85 & 72 & 181 & 91\\
        
         \hline
    \end{tabular}}
    \vspace{-1mm}
    \label{tab:performanceComparisionWithOthers}
\end{table}


\begin{figure*}[ht] 
\footnotesize
\centering
\subfigure[{Hurricane with $\epsilon=$ 1E-2}]
{ 
\raisebox{-1cm}{\includegraphics[scale=0.26]{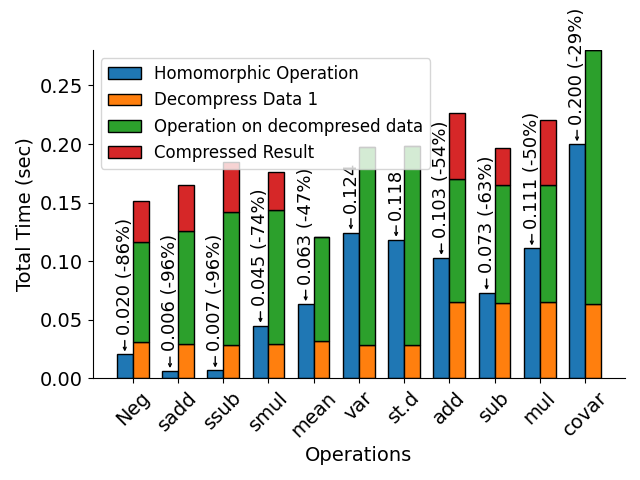}}%
}
\subfigure[{CESM-ATM  with $\epsilon=$ 1E-2}]
{
\raisebox{-1cm}{\includegraphics[scale=0.26]
{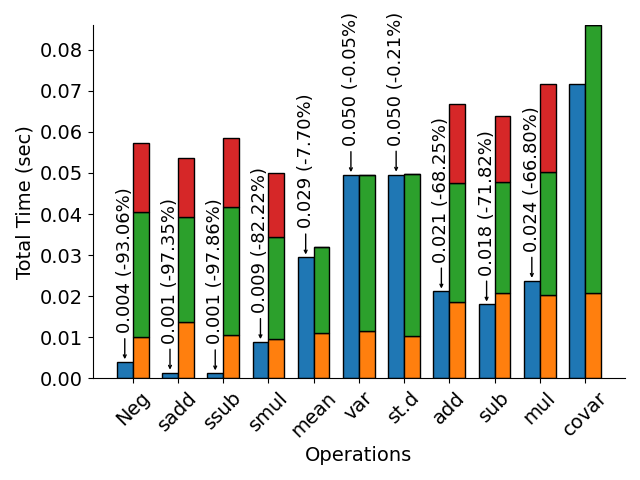}}%
}
\subfigure[{SCALE\_LETKF  with $\epsilon=$ 1E-2}]
{
\raisebox{-1cm}{\includegraphics[scale=0.26]{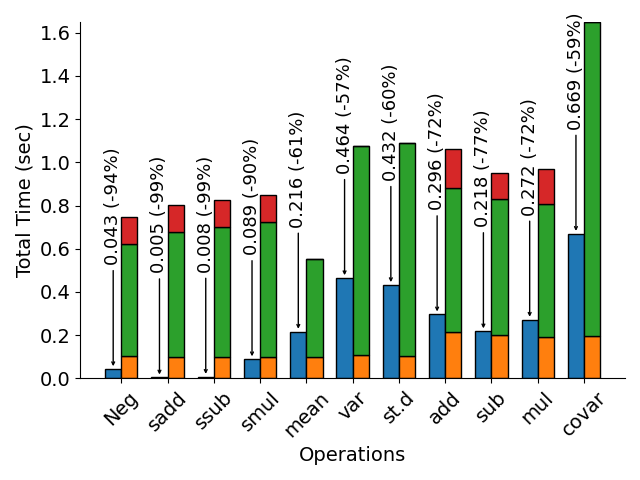}}%
}
\subfigure[{Miranda  with $\epsilon=$ 1E-2}]
{
\raisebox{-1cm}{\includegraphics[scale=0.26]{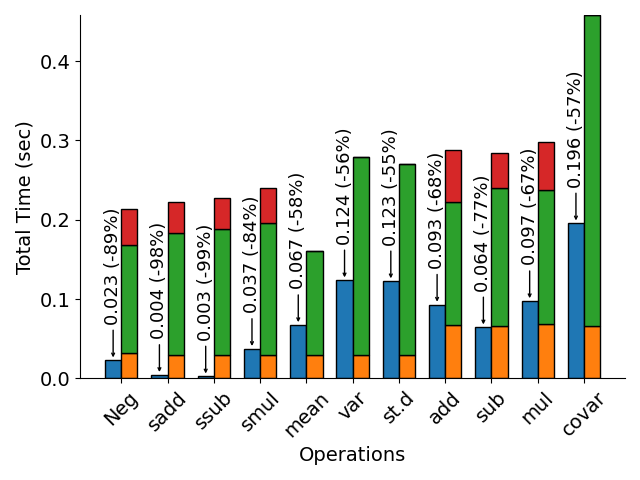}}%
}
\subfigure[{Hurricane  with $\epsilon=$ 1E-4}]
{
\raisebox{-1cm}{\includegraphics[scale=0.26]{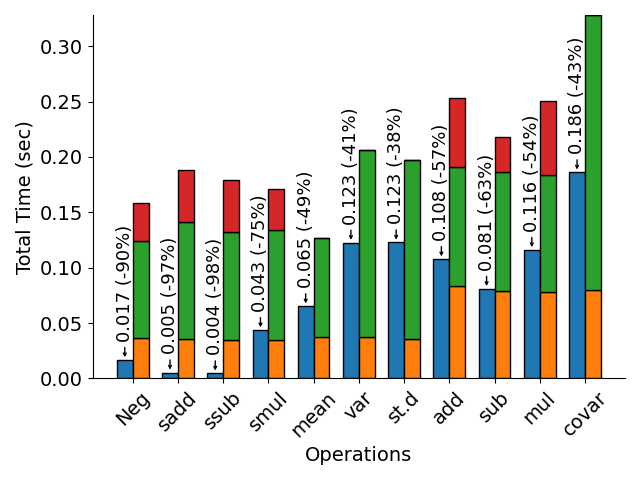}}%
}
\subfigure[{CESM-ATM  with $\epsilon=$ 1E-4}]
{
\raisebox{-1cm}{\includegraphics[scale=0.26]{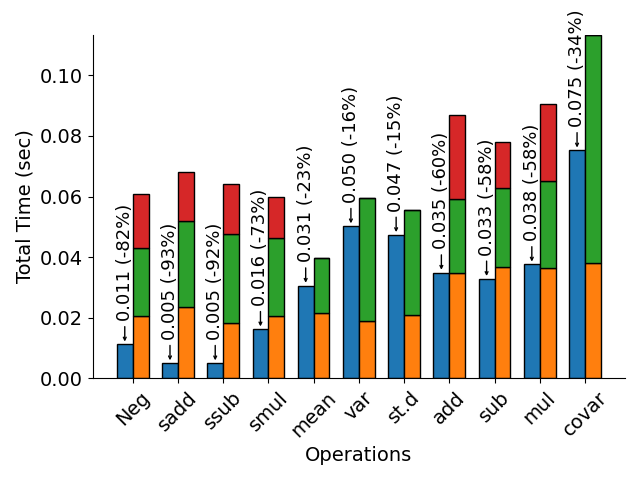}}%
}
\subfigure[{SCALE\_LETKF with $\epsilon=$ 1E-4}]
{
\raisebox{-1cm}{\includegraphics[scale=0.26]{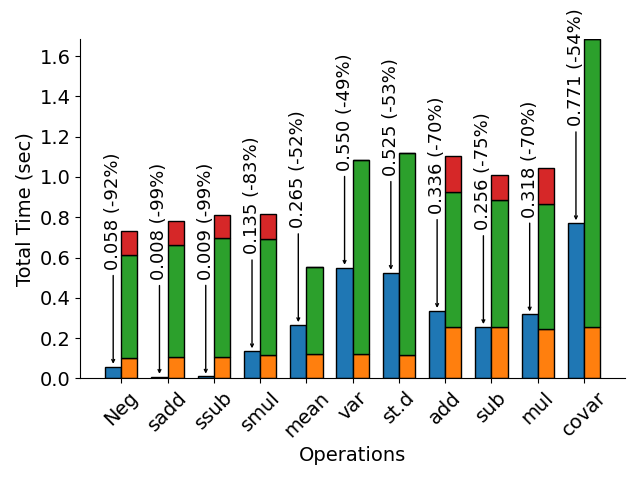}}%
}
\subfigure[{Miranda  with $\epsilon=$ 1E-4}]
{
\raisebox{-1cm}{\includegraphics[scale=0.26]{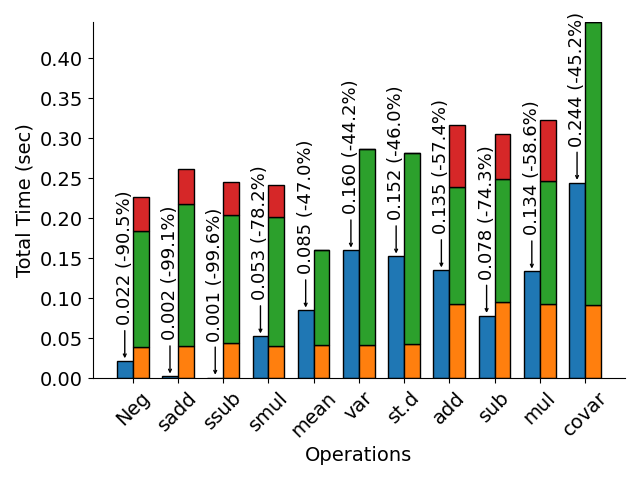}}%
}
\vspace{-1mm}
\caption{The time cost of various operations, including Decompression (orange), Operation (green), and Compression (red) times for SZp, as well as the total time (blue) for HoSZp, is compared using absolute error bounds ($\epsilon$) of 1E-2 (a-d) and 1E-4 (e-h). The total time of HoSZp encompasses the kernel time taken by different operations, including partial decompression and partial compression time taken by certain operations, as detailed in Section \ref{sec:HomoComp}. Each bar is color-coded to represent the time taken for a specific operation, as demonstrated in (a). \textbf{$<$time taken$>$ (- value \%)} on each blue bar represents the time taken by HoSZp operation and the percentage decrease in HoSZp's operation time in comparison to the corresponding SZp's operation time, respectively, for different datasets.}
\label{fig:eb-2-4Time}
\vspace{-3mm}
\end{figure*}

\begin{figure*}[ht] 
\centering
\hspace{-6mm}
\subfigure[{Hurricane   with $\epsilon=$ 1E-4}]
{
\raisebox{-1cm}{\includegraphics[scale=0.27]{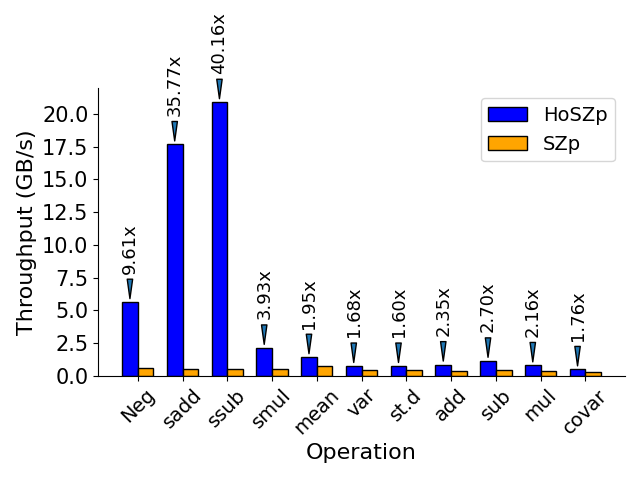}}%
}
\hspace{-1mm}
\subfigure[{CESM-ATM   with $\epsilon=$ 1E-4}]
{
\raisebox{-1cm}{\includegraphics[scale=0.27]{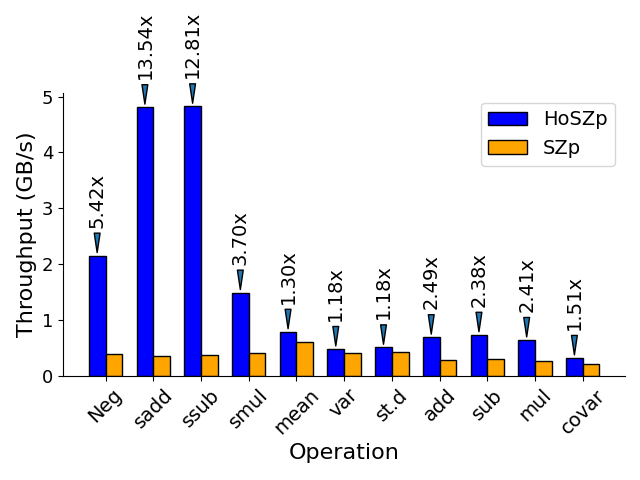}}%
}
\hspace{-4mm}
\subfigure[{SCALE\_LETKF   with $\epsilon=$ 1E-4}]
{
\raisebox{-1cm}{\includegraphics[scale=0.27]{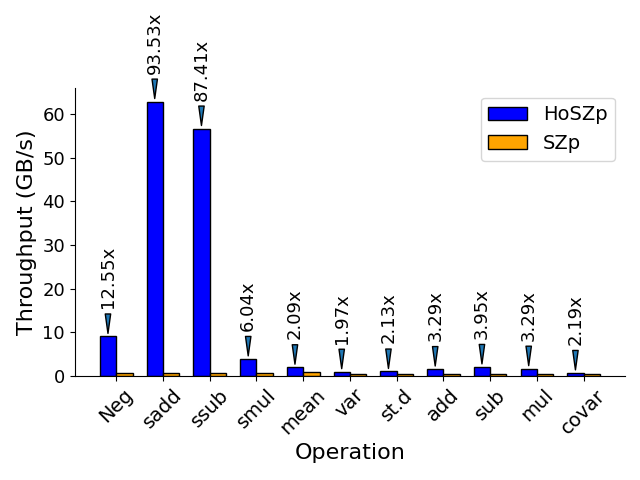}}%
}
\hspace{-4mm}
\subfigure[{Miranda   with $\epsilon=$ 1E-4}]
{
\raisebox{-1cm}{\includegraphics[scale=0.27]{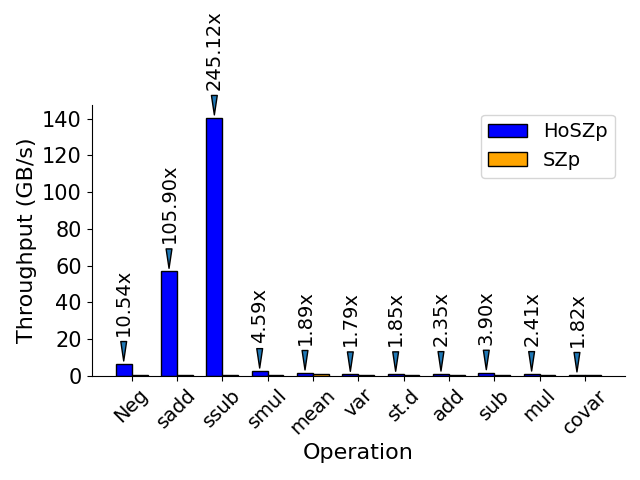} }%
}
\hspace{-8mm}
\caption{Kernel throughput for HoSZp and end-to-end throughput for SZp using absolute error bound ($\epsilon$) 1E-4.  The performance throughput ratio of each HoSZp operation with respect to SZp is shown above each blue bar.}
\label{fig:eb-4Throughput}
\vspace{-5mm}
\end{figure*}
\subsubsection{Time Cost}
\label{sec:TimeCalculation}
We evaluate the runtime for each operation in HoSZp and compare the results with SZp. For HoSZp, the time taken by each operation is calculated for the four datasets (Table \ref{tab:datasets}). For SZp, we perform the following tests: 
\begin{itemize}
    \item For scalar operation: decompression of compressed data + operations + compression.
    \item For scalar reduction, decompression + operation.
    \item For bivariate operation, decompression of compressed dataset 1 + decompression of compressed dataset 2 + operation + compression.
    \item For bivariate reduction, decompression of compressed dataset 1 + decompression of compressed dataset 2 + operation.
\end{itemize}
The time taken by each step in SZp is added to obtain the end-to-end time each operation takes. We measure the time for each field of the four datasets using two absolute error bounds of 1E-2 and 1E-4. Note that dataset 1 and dataset 2 used for bivariate operations and reduction are the same. Hence, we add the decompression time of both datasets to show the analysis later in Figure \ref{fig:eb-2-4Time}.

We observe in Figure \ref{fig:eb-2-4Time} that the time cost by HoSZp operations (blue color) is significantly lower (except for \textit{computation-as-output} for specific examples) than the time taken by operations performed using SZp, i.e., the total time cost on decompression, operation, and compression steps (shown with orange, green, and red colors). The time required for univariate and bivariate operations \textit{(compression-as-output)} is less as compared to the traditional SZp pipeline. This efficiency is attributed to the utilization of kernel operations from HoSZp, which involve either partial decompression or no decompression at all for certain operations such as negation, scalar addition, and scalar subtraction (see table \ref{tab:performanceReason}). As a result, the overhead of decompression and subsequent compression time are substnatially reduced or even completely eliminated.
We also observe that the reduction operation (\textit{computation-as-output}) applied to the compressed data takes less time than the traditional method of SZp in the Hurricane, SCALE-LETKF, and Miranda datasets. However, the reduction operations for CESM-ATM data with $\epsilon$ = 1E-2 (as depicted in Figure \ref{fig:eb-2-4Time}) do not show a significant reduction compared to the time required by SZp's reduction operations. This difference arises from the prominent dependence of reduction operations on the number of constant blocks (see Table \ref{tab:constantblocks}). These constant blocks contain zero values, enabling their exclusion during computation. Consequently, the accrued time savings are relatively minimal for datasets such as CESM-ATM with $\epsilon$ = 1E-2.
Additionally, it is noteworthy that the reduction operations exhibit comparatively improved performance in CESM-ATM data utilizing $\epsilon$ = 1E-4 despite the reduced number of constant blocks. This discrepancy arises because the performance depends on the floating-point ratio to integer computations, as summarized in Table \ref{tab:FPvsInt}. Notably, the floating-point values, obtained through SZp's decompression with $\epsilon$ = 1E-4, inherently possess higher precision, and hence they require more computation time than HoSZp's integer computations.
Note performing HoSZp reduction operations might not always be faster than the traditional method. However, it still saves memory as we do not have to perform complete decompression and store all the data in the memory to perform reduction operations.

\subsubsection{Throughput Analysis}
\label{sec:ThroughputCalculation}
We evaluate the throughput of HoSZp and compare it with SZp for different operations. We measure the end-to-end throughput of  SZp for each operation using the absolute error-bound 1E-4 for each field of the four datasets. We also evaluate the kernel throughput of HoSZp using the same absolute error bound and compare it with the end-to-end throughput of SZp. 
\begin{table}
\footnotesize
    \centering  
    \caption{Reasons to performance improvement for different operations}
    \vspace{-1mm}
    \begin{tabular}{|c|c|}
    \hline
        \textbf{Operations} &  \textbf{Reason}
        \\
        \hline  \hline      
        \multirow{2}*{\textbf{Scalar operations}} & No decompression\\
        &(partial decompression 
        + constant blocks \\
         &  only for scalar multiplication) \\ \hline 
        \textbf{Scalar Reductions} & constant blocks + integer data operations\\ \hline
        \multirow{2}*{\textbf{Bivariate Operations}} & Partial decompression + constant blocks \\
         & + integer data computation\\ \hline
        \textbf{Bivariate Reductions} & constant blocks + integer data computation\\
        
         \hline
    \end{tabular}
    \vspace{-1mm}
    \label{tab:performanceReason}
\end{table}
We observe in Figure \ref{fig:eb-4Throughput} that the throughput of HoSZp (shown with navy blue color) is usually higher than the end-to-end throughput of SZp (shown with yellow color). This is because we execute all operations within fully or partially compressed spaces while also excluding constant block computations. As a result, time is saved during the data decompression. Hence, more data can be processed per unit of time, increasing the throughput of HoSZp. The throughput of reduction operations (\textit{computation-as-output}) is lowest amongst other scalar and bivariate operations because the reduction operations are dependent on the dataset, i.e., constant and non-constant block (as explained previously).
\begin{table}
\footnotesize
    \centering  
    \caption{Total blocks and constant blocks in each dataset over all the fields for error bound $(\epsilon)$ 1E-2.} 
    \vspace{-1mm}
    \begin{tabular}{|c|c|c|c|c|}
    \hline
        \textbf{Datasets} 
        & \textbf{Const. blocks} &  \textbf{Total blocks} & \textbf{\% (Const./Total)} \\
        \hline \hline       
        \textbf{Hurricane}  & 360827  & 2734375 & 13\% \\
        \hline
        \textbf{CESM-ATM}  & 7817 & 506250 & 1.5\%\\
        \hline
        \textbf{SCALE-LETKF} & 1071863 & 26460000 & 4\%\\
        \hline
        \textbf{Miranda}  &593722 &4128768 & 14\%\\
         \hline
    \end{tabular}
    \vspace{-1mm}
    \label{tab:constantblocks}
\end{table}
\begin{table}
\footnotesize
    \centering  
    \caption{Total time (in seconds) of mean calculation cost on integer vs. floating-point with error bound $(\epsilon)$ 1E-2.} 
    \vspace{-1mm}
    \begin{tabular}{|c|c|c|c|c|}
    \hline
        \textbf{Datasets} &  \textbf{Int comp.}
        & \textbf{FP comp.} & \textbf{\% (Int/FP)}  \\
        \hline \hline       
        \textbf{Hurricane} & 0.46 & 0.65 & 71\%\\
        \hline
        \textbf{CESM-ATM} & 0.08 & 0.13 & 61\%\\
        \hline
        \textbf{SCALE-LETKF} & 4.72& 5.81 & 81\%\\
        \hline
        \textbf{Miranda} &0.71 &0.94 & 75\%\\
         \hline
    \end{tabular}
    \vspace{-1mm}
    \label{tab:FPvsInt}
\end{table}

\subsubsection{Compression Ratio}
\label{sec:compressionRatio}
In Table \ref{tab:compressionRatio}, we evaluate the compression ratios for different compressors using an absolute error bound of 1E-4. HoSZp outperforms SZp in terms of compression ratio but falls behind SZ, SZ3, and ZFP. The higher compression ratios of SZ, SZ3, and ZFP can be attributed to their advanced data decorrelation techniques, such as dynamic interpolation and orthogonal transform, and their effective lossless-encoding methods, such as Huffman/Zstd \cite{zstd} and embedded coding \cite{zfp}. 
The HoSZp may have a higher compression ratio than SZp does, mainly because HoSZp reorganizes the outliers in the pipeline (see Figure~\ref{fig:pipeline}), making the linear recurrence decoding steps for combining the compressed data for each block easier. 
This improvement eliminates the need to store compressed byte length limits per block, a significant limitation in SZp's compression efficiency \cite{SZp}. It is worth noting that although HoSZp has lower compression ratios than other modern compressors such as SZ, SZ3 and ZFP, it exhibits substantially higher  throughputs on various operations (see Figure \ref{fig:eb-4Throughput} and Table \ref{tab:performanceComparisionWithOthers}): 2$\times$-245$\times$ in most of cases, which is critical to the online execution performance of large-scale scientific applications. 

\begin{table}[ht]
\vspace{-2mm}
    \centering  
    \caption{Average compression ratios for different scientific simulation data using different compressors.}
    \vspace{-1mm}
\resizebox{0.98\columnwidth}{!}{%
    \begin{tabular}{|c|c|c|c|c|c|c|c|}
    \hline
        \textbf{Datasets} &  \textbf{HoSZp}& \textbf{SZp} &  \textbf{SZ} & \textbf{SZ3} & \textbf{SZx} & \textbf{ZFP} \\
        \hline     \hline   
        \textbf{Hurricane} & 2.78 & 1.59 &8.83 & 10 & 3.6 & 4.4\\
         \hline   
        \textbf{CESM-ATM} & 2.68 & 2.33 &6.48 & 5.0& 2.17 & 3.01\\
         \hline   
         \textbf{SCALE-LETKF} & 17.02 & 15.21 & 360.65 & 205.74 & 37.13 & 69.48\\
          \hline   
          \textbf{Miranda}  & 6.19  & 4.97 & 24.64 & 27.70 & 5.11 & 8.78\\
         \hline
    \end{tabular}}
    \vspace{-1mm}
    \label{tab:compressionRatio}
\end{table}

\subsubsection{Data Visualization}
We evaluate the data quality of the operated data obtained using HoSZp and SZp by visualization.
In this experiment, we used the precipitation dataset of Hurricane Data for timestep 44 (Dataset 1) and timestep 48 (Dataset 2) and compressed it using the absolute error bound ($\epsilon$) of 1E-1. Then, we test HoSZp's element-wise addition operation on the compressed data of the two datasets 
\begin{figure}[ht] 
\footnotesize
\centering
\subfigure[{Decompressed dataset 1}]
{
\raisebox{-1cm}{\includegraphics[scale=0.21]{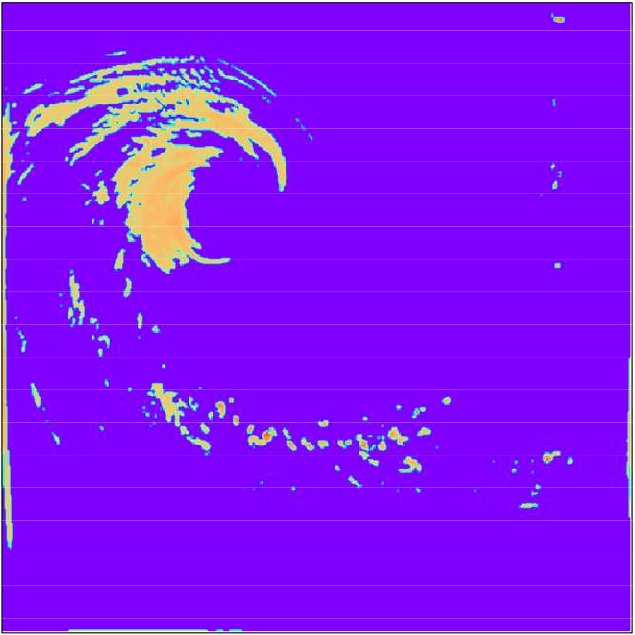}}%
}
\hspace{0.5mm}
\subfigure[{Decompressed dataset 2}]
{
\raisebox{-1cm}{\includegraphics[scale=0.21]{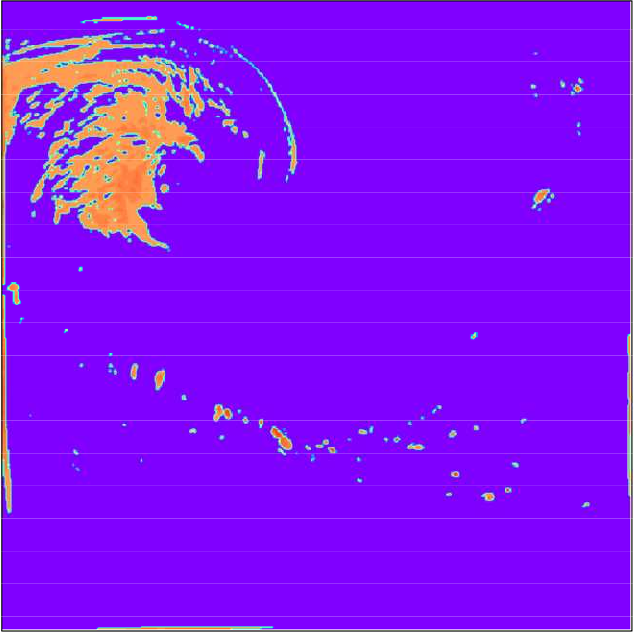}}%
}
\subfigure[{Traditional addition}]
{
\raisebox{-1cm}{\includegraphics[scale=0.21]{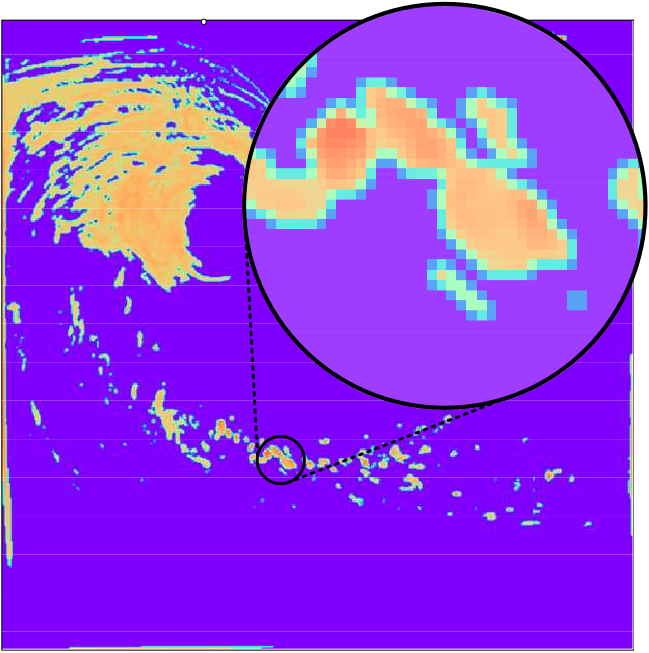}}%
}
\subfigure[{Homomorphic addition}]
{
\raisebox{-1cm}{\includegraphics[scale=0.21]{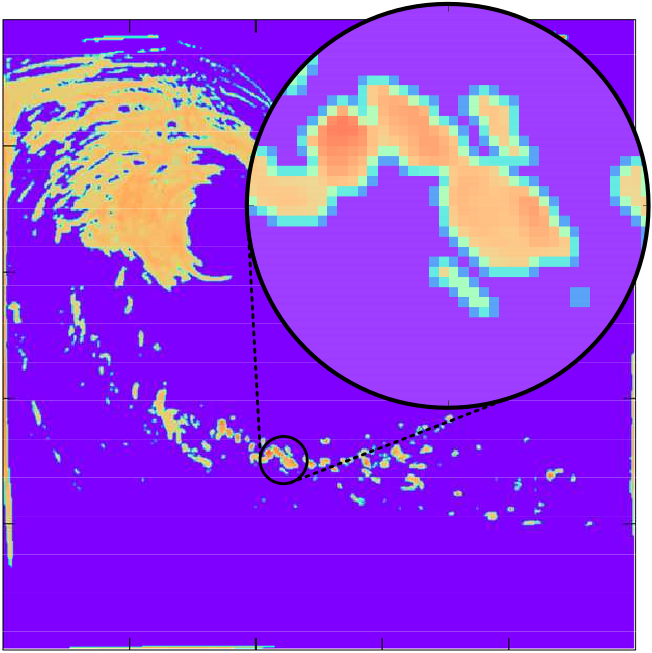}}%
}
\vspace{-3mm}
\caption{Visualization of Decompressed Hurricane Precipitation Dataset for (a) timestep 44 and (b) timestep 48 in log domain for slice number 38. The result of adding (a) and (b) using traditional workflow is shown in (c), and the result of adding the compressed data of (a) and (b) homomorphically is shown in (d). Experiments are performed with error bound $\epsilon =$ 1E-4.}
\label{fig:Viseb-1}
\end{figure}
(i.e., Dataset 1 + Dataset 2), and then perform decompression to obtain the final output. For a comparison, we also perform the same addition operation using the traditional workflow of SZp where we first decompress the two compressed datasets and then apply element-wise addition.

The visualization results are presented in Figure \ref{fig:Viseb-1}. Figure \ref{fig:Viseb-1}(a) and \ref{fig:Viseb-1}(b) shows the decompressed dataset obtained by applying the SZp workflow. Figure \ref{fig:Viseb-1}(c) shows the SZp's traditional workflow result, and \ref{fig:Viseb-1}(d) shows HoSZp's result.
We observe that Figure \ref{fig:Viseb-1}(c) and \ref{fig:Viseb-1}(d) are identical to each other. This is because all the operations in HoSZp do not perform decompression of quantization (\textbf{Q}), and hence no further data loss may happen in the HoSZp's operations when compared with the traditional method of performing the operations (also see Theorem \ref{thm:1}). This experiment clearly shows that HoSZp operations can retain the data quality and are homomorphic to the operation when applied using traditional workflow.
\label{sec:DataVisualization}

\subsubsection{Distributed Experiments}

Considering compression techniques are widely used in distributed systems~\cite{huang2024optimizing, huang2023ccoll, huang2023gzccl}, we perform a distributed experiment based on Reverse Time Migration (RTM) simulation data with up to 16 nodes, in order to validate the effectiveness of HoSZp. RTM \cite{RTM-tutorial} is a typical parallel application that may generate vast amount of data to be kept in memory \cite{Huang-RTM-ccgrid23}, which will be used in the calculations later. The traditional solutions  transfer the data from memory to disks and load them back to memory when needed, suffering a significantly high overhead. Latest solutions such as \cite{Huang-RTM-ccgrid23} maintain the lossy-compressed data in memory, and fully decompress them when needed in the later simulation. We evaluate HoSZp vs. SZp based on this application. Specifically, each execution node compresses the dataset and sends the compressed data to a root node. The root node decompresses the received data and executes the aggregation (sum) operation. Table \ref{tab:MPI-HoSZp} shows our HoSZp is able to achieve up-to 2.08$\times$ speedups compared with the traditional SZp. The performance gain increases with error-bounds, because the larger the error-bound is, the smaller the compressed data size will be, which benefits our HoSZp more compared with the traditional SZp.


\begin{table}[ht]
\vspace{-2mm}
\centering
\caption{Speed-ups of HoSZp over SZp in a distributed experiment.}
\label{tab:MPI-HoSZp}
\vspace{-1mm}
\resizebox{0.6\columnwidth}{!}{%
\begin{tabular}{|c|c|c|c|}
\hline
\textbf{ABS} & \textbf{4 nodes} & \textbf{8 nodes} & \textbf{16 nodes} \\ \hline\hline
\textbf{1E-5} & 1.66$\times$ & 1.82$\times$ & 1.95$\times$ \\ \hline
\textbf{1E-4} & 1.71$\times$ & 1.86$\times$ & 2.02$\times$ \\ \hline
\textbf{1E-3} & 1.77$\times$ & 1.91$\times$ & 2.08$\times$ \\ \hline
\end{tabular}%
}
\vspace{-1mm}
\end{table}

\section{Conclusion and Future Work}
We propose HoSZp, an error-bounded lossy compressor that can perform homomorphic operations in compressed space. HoSZp consists of a bunch of lightweight scalar and multi-variate operations. We perform experiments with real-world datasets, and the key findings are summarized as follows:
\begin{itemize}
    \item HoSZp can achieve higher throughput than a SZp while providing reasonable compression ratios.
    \item HoSZp features guaranteed error control with identical operation results in the compressed data compared with traditional full-decompression-dependent workflow. 
    \item We perform an experiment using a real-world RTM dataset to show the effectiveness of using HoSZp in a distributed environment, with  2.08$\times$ performance speedup.
\end{itemize}
In the future, we will  add more homomorphic operations, like distance measures, similarity measures, and homomorphic compositions to make the tool more powerful. 
\label{sec:conclude}

\section*{Acknowledgments}
This research was supported by the U.S. Department of Energy, Office of Science, Advanced Scientific Computing Research (ASCR), under contract DE-AC02-06CH11357. The experimental resource for this paper was provided by the Laboratory Computing Resource Center on the Bebop cluster at Argonne National Laboratory. 

We acknowledge the computing resources provided on Bebop (operated by Laboratory Computing Resource Center at Argonne) and ACCESS Anvil machine operated by Purdue University.

\bibliographystyle{IEEEtran}
\bibliography{references}

\end{document}